\documentclass[a4paper,10pt]{amsart}

\usepackage[english]{babel}
\usepackage[T1]{fontenc}
\usepackage[latin1]{inputenc}
\usepackage{xspace}
\usepackage{hyperref}
\usepackage{subfigure}
\usepackage{color}
\usepackage{mathrsfs}
\usepackage{mathrsfs}
\usepackage{amsfonts}
\usepackage{amssymb}
\usepackage{amsmath}
\usepackage{array,multirow}
\usepackage{mathbbol}

\usepackage{color}
\usepackage{graphicx}
\usepackage{graphics}

\usepackage[lined, ruled,vlined,algo2e]{algorithm2e}
\usepackage{mathabx}
\def\cN{\mathcal{N}}
\def\cB{B}
\def\cA{\mathcal{A}}
\def\cC{\mathcal{C}}
\def\cI{I}
\def\cH{\mathcal{H}}
\def\cE{\mathcal{E}}

\def\cS{\mathcal{S}}

\def\cD{\mathcal{D}}

\def\bN{\mathbb{N}}

\def\sC{\mathscr{C}}
\def\sG{\mathscr{G}}
\def\sH{\mathscr{H}}

\def\dom{{\sc Dom-Enum}}
\def\transhyp{{\sc Trans-Enum}}

\def\TDS{{\sc TDom-Enum}}
\def\CDS{{\sc CDom-Enum}}

\newcommand{\eg}{\emph{e.g.}\xspace}
\newcommand{\ie}{\emph{i.e.}\xspace}

\newcommand{\macro}[3]{\newcommand{#1}[#3]{#2}}
\macro{\size}{\|#1\|}{1}

\newtheorem{thm}{Theorem}
\newtheorem{lem}[thm]{Lemma}
\newtheorem{defn}[thm]{Definition}
\newtheorem{cor}[thm]{Corollary}
\newtheorem{prop}[thm]{Proposition}
\newtheorem{Claim}[thm]{Claim}
\newtheorem{rem}[thm]{Remark}

\title[The \dom\ Problem and Related Notions]{On the Enumeration of Minimal Dominating Sets and Related Notions}
\author[M.M. Kanté \and V. Limouzy\and A. Mary \and L. Nourine]{Mamadou Moustapha Kant\'e \and Vincent Limouzy \and Arnaud
  Mary \and Lhouari Nourine}

\address{Clermont-Universit\'e, Universit\'e Blaise Pascal, LIMOS,
  CNRS, France}
\email{\{mamadou.kante,limouzy,mary,nourine\}@isima.fr}
\thanks{M.M. Kanté and V. Limouzy are supported by the French Agency
  for Research under the DORSO project.}
\thanks{A. Mary and L. Nourine are partially supported by the
  French Agency for Research under DEFIS program DAG, ANR-09-DEFIS,
  2009-2012.}
\thanks{A preliminary version of Sections 4 and 5 appeared in the proceedings of FCT 2011.}
\begin{document}

 \begin{abstract} 
A dominating set $D$ in a graph is a subset of its vertex set such
that each vertex is either in $D$ or has a neighbour in $D$. In this
paper, we are interested in the enumeration of (inclusion-wise) minimal
dominating sets in graphs, called the \dom\ problem. It is well known
that this problem can be polynomially reduced to the \transhyp\ problem in
hypergraphs, i.e., the problem of enumerating all minimal transversals
in a hypergraph.  Firstly we show that the \transhyp\ problem can be
polynomially reduced to the \dom\ problem. As a consequence there
exists an output-polynomial time algorithm for the \transhyp\ problem
if and only if there exists one for the \dom\ problem.  Secondly, we
study the \dom\ problem in some graph classes.  We give an
output-polynomial time algorithm for the \dom\ problem in split graphs, and
 introduce the completion of a graph to obtain an output-polynomial time
algorithm for the \dom\ problem in $P_6$-free chordal graphs, a proper
superclass of split graphs. Finally, we investigate the complexity of
the enumeration of (inclusion-wise) minimal connected dominating sets and minimal total dominating sets
of graphs. We show that there exists an output-polynomial time
algorithm for the \dom\ problem (or equivalently \transhyp\ problem) if
and only if there exists one for the following enumeration problems: minimal total dominating sets, minimal total
dominating sets in split graphs,
minimal connected dominating sets in split graphs, minimal dominating sets in
co-bipartite graphs.

 \end{abstract}

\maketitle
\section{Introduction}\label{sec:1}

The \textsc{Minimum Dominating Set} problem is a classic and
well-studied graph optimisation problem. \emph{A} \emph{dominating
  set} in a graph $G$ is a subset $D$ of its set of vertices such that
each vertex is either in $D$ or has a neighbour in $D$. Computing a
minimum dominating set has numerous applications in many areas, \eg,
networks, graph theory (see for instance the book \cite{HHS98}). In this
paper we are interested in the enumeration of \emph{minimal
  (connected, total) dominating sets} in graphs.

Enumeration problems have received much interest over the past
decades due to their applications in computer science
\cite{AgrawalIS93,DiehlJR93,wille,GunopulosKMT97,ruskey}.  For these
problems the size of the output may be exponential in the size of the
input, which in general is different from optimisation or counting
problems where the size of the output is polynomially related to the
size of the input. A natural parameter for measuring the time
complexity of an enumeration algorithm is the sum of the sizes of the input and
output. An algorithm whose running time is bounded by a polynomial
depending on the sum of  the sizes of the input and output is called an
\emph{output-polynomial time} algorithm (also called
\emph{total-polynomial time} or \emph{output-sensitive} algorithm).

The enumeration of minimal dominating sets of graphs (\dom\ problem for short) is closely related to the well-known {\sc Trans-Enum} problem in hypergraphs, which consists in enumerating the set of
minimal \emph{transversals} (or \emph{hitting sets}) of a hypergraph.  A transversal of a hypergraph is a subset of its ground set which has a non empty intersection with every hyperedge. One can
notice that the set of minimal dominating sets of a graph is in bijection with the set of minimal transversals of its \emph{closed neighbourhood hypergraph} \cite{BLS99}.  The {\sc Trans-Enum} problem
has been intensively studied due to its connections to several problems in such fields as data-mining and learning \cite{EG95,EGM03,GunopulosKMT97,BEKG06,nourine12}.  It is still open whether there
exists an output-polynomial time algorithm for the {\sc Trans-Enum} problem, but several classes where an output-polynomial time algorithm exists have been identified (see for instance the survey
\cite{EGM08}). So, classes of graphs whose closed neighbourhood hypergraphs are in one of these identified classes of hypergraphs admit also output-polynomial time algorithms for the \dom\
problem. Examples of such graph classes are planar graphs and bounded degree graphs (see \cite{KLMN11,KLMN12} for more information).  Recently, the \dom\ problem has been studied by several groups of
authors \cite{CouturierHHK13,FGPS08}. Their research on exact exponential-time algorithms triggered a new approach to the design of enumeration algorithms which uses classical worst-case running time
analysis, i.e., the running time depends on the length of the input.

In this paper, we first prove that the {\sc Trans-Enum} problem can be
polynomially reduced to the \dom\ problem. Since the other direction
also holds, the two problems are \emph{equivalent}, \ie, there exists
an output-polynomial time algorithm for the \dom\ problem if and only
if there exists one for the {\sc Trans-Enum} problem. One could possibly
expect to benefit from graph theory tools to solve the two problems
and at the same time many other enumeration problems equivalent to the
{\sc Trans-Enum} problem (see \cite{EG95} for examples of problems
equivalent to {\sc Trans-Enum}).  In addition, we show that there exists an output-polynomial time algorithm for
the \dom\ problem (or equivalently \transhyp\ problem) if and only if
there exists one for the following enumeration problems: \TDS\ problem,
\CDS\ in split graphs, \TDS\ in split graphs, \dom\ in co-bipartite
graphs, where the \TDS\ problem corresponds to the enumeration of minimal \emph{total dominating} sets.

We then characterise graphs where the addition of edges changes the set of minimal dominating sets. The
maximal extension (addition of edges) that keeps invariant the set of
minimal dominating sets can be computed in polynomial time, and appears
to be a useful tool for getting output-polynomial time algorithms for
the \dom\ problem in new graph classes such as $P_6$-free chordal
graphs. As a consequence, \dom\ in split graphs  and \dom\ in $P_6$-free chordal graphs are  linear delay and polynomial space.
%
%
%

We finally study the complexity of the enumeration  of minimal \emph{connected dominating} sets (called the \CDS\ problem). 
The \textsc{Minimum Connected Dominating Set} problem is a well-known and well-studied variant of the \textsc{Minimum Dominating Set} problem due to its
 applications in networks \cite{HHS98,WL01}.  We have proved in \cite{KLMN11} that \CDS\ in split graphs is
equivalent to the {\sc Trans-Enum} problem. We will extend this result to
other graph classes. Indeed, we prove that the minimal
connected dominating sets of a graph are the minimal transversals of
its minimal separators. As a consequence, in any class of graphs with
a polynomially bounded number of minimal separators, the \CDS\ problem
can be polynomially reduced to the {\sc Trans-Enum} problem; examples of
such classes are chordal graphs, circle graphs and circular-arc
graphs \cite{BKKM98,KK95,KKW98}. Finally, we show that the \CDS\ problem is harder than the \dom\ problem.

\medskip

\paragraph{\bf Paper Organisation.} Some needed definitions are defined in Section
\ref{sec:2}. The equivalence between the {\sc Trans-Enum} problem, the
\dom\ problem and the {\sc TDom-Enum} problem is given in Section \ref{sec:3}. We recall in Section
\ref{sec:4} the output-polynomial time algorithm for the \dom\ problem
in split graphs published in \cite{KLMN11}.  Maximal extensions
(additions of edges) of graphs are defined in Section \ref{sec:6} and
a use of these maximal extensions to obtain an output-polynomial time
algorithm for the \dom\ problem in $P_6$-free chordal graphs is also
given. The \CDS\ problem is investigated in Section \ref{sec:7}.

\section{Preliminaries}\label{sec:2}

If $A$ and $B$ are two sets, $A\setminus B$ denotes the set $\{x\in
A\mid x\notin B\}$. The power-set of a set $V$ is denoted by $2^V$.
We denote by $\bN$ the set containing zero and the positive
integers. The size of a set $A$ is denoted by $|A|$. 

We refer to \cite{Diestel2005} for graph terminology not defined
below; all graphs considered in this paper are undirected, finite and
simple.  A graph $G$ is a pair $(V(G),E(G))$, where $V(G)$ is the set
of vertices and $E(G)\subseteq V(G)\times V(G)$, the set of edges, is
symmetric.  An edge between $x$ and $y$ is denoted by $xy$
(equivalently $yx$).  The subgraph of $G$ induced by $X\subseteq
V(G)$, denoted by $G[X]$, is the graph $(X,E(G)\cap (X\times X))$;
$G\setminus X$ is the graph $G[V(G)\setminus X]$. A graph is said to
be \emph{chordal} if it has no induced cycle of length greater than or
equal to $4$; it is a \emph{split} graph if its vertex set can be
partitioned into an independent set $S$ and a clique $C$. Notice that
split graphs form a proper subclass of chordal graphs. For two graphs
$G$ and $H$, we say that $G$ is \emph{H-free} if $G$ does not contain
$H$ as an induced subgraph.  For $k\geq 1$, we let $P_k$ be the path
on $k$ vertices. For a graph $G$, we let $N_G(x)$, the set of
neighbours of $x$, be the set $\{y\in V(G)\mid xy\in E(G)\}$, and we
let $N_G[x]$ be $N_G(x)\cup \{x\}$. For $X\subseteq V(G)$, we write
$N_G[X]$ and $N_G(X)$ for respectively $\bigcup\limits_{x\in X}
N_G[x]$ and $N_G[X]\setminus X$.

A \emph{dominating set} in a graph $G$ is a set of vertices $D$ such that every
vertex of $G$ is either in $D$ or is adjacent to some vertex of $D$. It is said
to be \emph{minimal} if  it does not contain any other dominating set as a subset.
The set of all minimal dominating sets of $G$ will be denoted
by $\mathcal{D}(G)$.  Let $D$ be a dominating set of $G$ and $x \in D$. We say
that $x$ has a \emph{private neighbour} $y$ in $G$ if $y\in N_G[x]\setminus
N_G[D\setminus \{x\}]$. Note that a private neighbour of a vertex $x\in
	D$ in $G$ is either $x$ itself, or a vertex in $V(G)\setminus D$, but never a vertex
	$y\in D\setminus \{x\}$. The set of private neighbours of $x\in D$ in $G$ is
denoted by $P_D(x)$. The following is straightforward.

\begin{lem} \label{lem:2.1}
Let $D$ be a dominating set of a graph $G$. Then  $D$ is a minimal dominating set if and only if  $P_D(x)\neq \emptyset$ for every
$x\in D$.
\end{lem}

%

A {\em hypergraph} $\mathcal{H}$ is a pair $(V(\cH),\cE(\cH))$ where $V(\cH)$
is a finite set and $\cE(\cH) \subseteq 2^{V(\cH)}\setminus \{ \emptyset \}$. It is worth noticing that graphs
are special cases of hypergraphs. We will call the elements of $V(\cH)$
vertices and elements of $\cE(\cH)$ hyperedges, and when the context is clear a
hypergraph will be denoted by its set of hyperedges only.  If $\cH$ is
a hypergraph, we let $\cI(\cH)$, the \emph{bipartite incidence} graph of $\cH$,
be the graph with vertex set $V(\cH)\cup \{y_e\mid e\in \cE(\cH)\}$ and
edge set $\{xy_e\mid x\in V(\cH),\ e\in \cE(\cH)$ and $x\in e\}$.  Note that
the neighbourhood of the vertex $y_e$ in $\cI(\cH)$ is exactly the set $e$. A
hypergraph $\cH$ is said to be \emph{simple} if
\begin{enumerate}
\item[(i)] for all $e,e'\in \cE(\cH)$,  $e\subseteq e' \Longrightarrow e=e'$,
	and 
\item[(ii)]  $V(\cH)=\bigcup\limits_{e\in \cE(\cH)} e$.
\end{enumerate}

For a hypergraph $\cH$ we denote by $Min(\cH)$ the hypergraph on the
	same vertex set and keeping only minimal hyperedges, \ie, $\cE(Min(\cH)):=\{
		e\in \cE(\cH) \mid \forall e'\in \cE(\cH)\setminus \{e\},~e'\not\subseteq e  \}$. A {\em transversal} (or {\em hitting set}) of
$\mathcal{H}$ is a subset of $V(\cH)$ that has a non-empty
	intersection with every hyperedge of $\cE(\cH)$; it is {\em minimal} if it
does not contain any other transversal as a subset. The set of all minimal
transversals of $\mathcal{H}$ is denoted by $tr(\mathcal{H})$.  The size of a
hypergraph $\cH$, denoted by $\size{\cH}$, is $|V(\cH)| +
\sum\limits_{e\in\cE(\cH)}|e|$. The set of all hypergraphs (respectively
	all graphs) is denoted by $\sH$ (respectively $\sG$). 

\begin{prop}[\cite{Berge1989}] \label{prop:2.1}
 For each simple hypergraph $\cH$, we have $tr(tr(\cH))=\cH$.
\end{prop}

From Proposition \ref{prop:2.1}, we obtain the following.

\begin{cor}\label{cor:2.1}
 For each simple hypergraph  $\cH$ and each $x\in V(\cH)$, there
 exists $T\in tr(\cH)$ such that $x\in T$. 
\end{cor}



An \emph{enumeration algorithm} (algorithm for short) for a set $\sC$ is
an algorithm that lists the elements of $\sC$ without repetitions. Let
$\varphi(X)$ be a hypergraph property where $X$ is a subset of vertices
(for instance $\varphi(X)$ could be ``$X$ is a transversal''). For
a hypergraph $\cH$, we let $\sC_\varphi(\cH)$ be the set $\{Z\subseteq
	V(\cH)\mid \varphi(Z)$ is true in $\cH\}$.  An \emph{enumeration problem}
for the hypergraph property $\varphi(X)$ takes as input a hypergraph
$\cH$, and the task is to enumerate, without repetitions, the set
$\sC_\varphi(\cH)$. An algorithm for $\sC_\varphi(\cH)$ is an
\emph{output-polynomial time} algorithm if there exists a polynomial
$p:\bN\to \bN$ such that $\sC_\varphi(\cH)$ is listed in time $p(||\cH||
+ ||\sC_\varphi(\cH)||)$. Notice that since an algorithm $\cA$ for an
enumeration problem takes a hypergraph as input and outputs
a hypergraph with same vertex set, we can consider it as
a function $\cA:\sH\to \sH$.
Let $f:\bN\to \bN$. We say that an algorithm enumerates $\sC_\varphi(\cH)$ with
delay $f(||\cH||)$ if, after a polynomial time pre-processing, it
outputs the elements of $\sC_\varphi(\cH)$  without repetitions, the
delay between two outputs being bounded by $f(||\cH||)$. If $f$ is a
polynomial (or a linear function), we call it a polynomial (or linear)
delay algorithm.

\begin{defn}
	Let $P$ and $P'$ be enumeration problems for hypergraph
		properties $\varphi(X)$ and $\varphi'(X)$ respectively.  We say
		that $P'$ is \emph{at least as hard} as $P$, denoted by
		$P \preceq_{op} P'$, if an output-polynomial time algorithm for
		$P'$ implies an output-polynomial time algorithm for $P$.
\end{defn}

Two enumeration problems $P$ and $P'$ are \emph{equivalent} if
$P\preceq_{op} P'$ and $P'\preceq_{op} P$.  We denote by
\transhyp\ the enumeration problem of minimal transversals in
hypergraphs. Similarly, we denote by \dom\ the enumeration problem of
minimal dominating sets in graphs. For a problem $P$ and a subclass
$\cC$ of instances of $P$, we denote by $P(\cC)$ the problem $P$
restricted to the instances in $\cC$.  For instance, \dom(split
graphs) denotes the problem of enumerating the set of minimal
dominating sets in split graphs.


\section{\dom\ is Equivalent to \transhyp\ }\label{sec:3} 

The fact that \dom\ $\preceq_{op}$ \transhyp\ can be considered
	folklore. Let us remind it for completeness. For a graph $G$, we let
$\cN(G)$, the \emph{closed neighbourhood hypergraph}, be $(V(G),\{N_G[x]\mid
	x\in V(G)\})$.

\begin{lem}[Folklore \cite{BLS99}] \label{lem:3.1} Let $G$ be a graph and
	$D\subseteq V(G)$. Then $D$ is a dominating set of $G$ if and only if $D$ 
	is a transversal of $\cN(G)$ if and only if $D$ is a transversal
		of $Min(\cN(G))$.
\end{lem} 
\begin{cor}\label{cor:leq} 
	\dom\ $\preceq_{op}$ \transhyp.
\end{cor} 
\begin{proof} 
	From Lemma \ref{lem:3.1}, we have that $tr(\cN(G))=\cD(G)$. Hence, if we
	have an output-polynomial time algorithm for \transhyp\, then we
	can use it to enumerate all minimal dominating sets of a graph in
	output-polynomial time.
\end{proof}

\begin{cor}
	\label{lemx}Let $G$ be a graph and $x\in V(G)$. Then there
  exists $D\in \cD(G)$ such that $x\in D$.
\end{cor}

\begin{proof} 
    Corollary of Lemma \ref{lem:3.1} and Corollary \ref{cor:2.1}.
    \end{proof}

We now prove that \transhyp\ $\preceq_{op}$ \dom. One may
	wonder whether with every hypergraph $\cH$ one can associate a graph
$G$ such that $\cD(G)=tr(\cH)$. However, the following result shows
	that such a reduction does not exist.  

\begin{prop}\label{one}
 For every function $f:\sH\to \sG$, there exists $\cH \in \mathscr{H}$
 such that $tr(\cH)\neq\cD(f(\cH))$. 
\end{prop}

\begin{proof}
Let $\cH$ be a simple hypergraph with $|V(\cH)|=|\cE(\cH)|=n$ and such
that $\cH$ is not the closed neighbourhood hypergraph of any
graph. Such a hypergraph exists (see for instance
\cite{BorosGZ08}). Now assume that there exists a
graph $G$ such that $\cD(G)=tr(\cH)$. Note that since each vertex of a
simple hypergraph belongs to at least one minimal transversal
(Corollary \ref{cor:2.1}), and since each vertex of a graph appears in
at least one minimal dominating set (Corollary \ref{lemx}), we have
$V(G)=V(\cH)$.  By Lemma \ref{lem:3.1},
$tr(\cH)=tr(\cN(G))=tr(Min(\cN(G)))$ and so $\cH=Min(\cN(G))$
(Proposition \ref{prop:2.1}).  Furthermore, $Min(\cN(G))\subseteq
\cN(G)$ and $|Min(\cN(G))|=|\cE(\cH)|=n=|\cN(G)|$ and so
$Min(\cN(G))=\cN(G)$.  We conclude that $\cH=\cN(G)$ and then $\cH$ is
the closed neighbourhood hypergraph of $G$, which contradicts the assumption.
\end{proof}

Despite the above result, we can polynomially reduce \transhyp\ to \dom.
 In order to prove this statement we introduce the \emph{co-bipartite incidence}
 graph  associated with every hypergraph $\cH$.
 \begin{defn}\label{co-bip}
Let $\cH$ be a hypergraph. We associate with $\cH$ a co-bipartite incidence graph $\cB(\cH)$, defined as follows:
\begin{itemize}
\item $V(\cB(\cH)):=V(\cI(\cH)) \cup \{v\}$ with $v\notin V(\cI(\cH))$,

\item $E(\cB(\cH)):= E(\cI(\cH)) \cup \{vx\mid x\in V(\cH)\} \cup \{xy\mid
	x,y\in V(\cH)\} \cup \{y_ey_{e'}\mid e,e'\in \cE(\cH)\}$.
\end{itemize} 
\end{defn}

In other words, $\cB(\cH)$ is obtained from $\cI(\cH)$
by adding a new vertex that is made adjacent to all vertices in
$V(\cH)$, and replacing the subgraph induced by $V(\cH)$ (resp.
$\{y_e\mid e\in \cE(\cH)\}$) by a clique on the same set;  see Figure
\ref{fig:3.1} for an illustration. The following is straightforward to
prove.

\begin{figure}[h!]
	\includegraphics[scale=1.3]{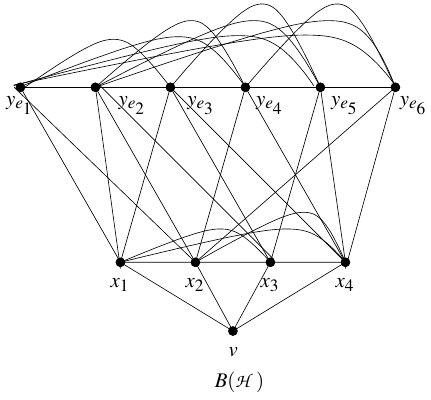} 
	\caption{An example of the
		co-bipartite incidence graph $\cB(\cH)$  of the hypergraph
		$\mathcal{H}=(\{x_1,x_2,x_3,x_4\},\{e_1,e_2,e_3,e_4,e_5,e_6\})$ where
		$e_1=\{x_1,x_2\}$, $e_2=\{x_1,x_2,x_3\}$, $e_3=\{x_1,x_3,x_4\}$,
		$e_4=\{x_2,x_4\}$, $e_5=\{x_3,x_4\}$, $e_6=\{x_2,x_4\}$. The set
		$\{x_1,x_2\}$ is a minimal transversal of $\cH$ and a minimal
		dominating set of $\cB(\cH)$. } \label{fig:3.1}
\end{figure}

\begin{lem}\label{lemme3}
Let $\cH$ be a hypergraph and $T$ a transversal of $\cH$. Then $T$ is a
dominating set of $\cB(\cH)$.
\end{lem}

The following lemma claims that there is only a quadratic number of minimal
dominating sets of $\cB(\cH)$ that are not minimal transversals of $\cH$. 

\begin{lem}\label{lemma2}
Let $\cH$ be a hypergraph and let $D$ be a minimal dominating set of
$\cB(\cH)$. Then $D$ is either equal to $\{x,y_e\}$ with $x \in V(\cH)\cup
\{v\}$ and $e\in \cE(\cH)$ or $D$ is a minimal transversal of $\cH$.  
\end{lem}
\begin{proof}
As $v$ must be dominated by $D$, $D\cap (V(\cH)\cup\{v\} )\neq
\emptyset$.  Let $x\in D\cap (V(\cH)\cup\{v\} )$.  Assume that $D\cap
\{y_e \mid e\in \cE(\cH)\}\neq \emptyset$. Since $D$ is a minimal
dominating set, $x$ dominates $V(\cH) \cup \{ v \}$, and since $\{ y_e
    \mid e \in \cE(\cH) \}$ is a clique, $|D\cap \{ y_e \mid e \in
    \cE(\cH) \}|=1$. This implies that $D$ is of the form $\{ x,y_e \}$.
So assume that $D\subseteq V(\cH)\cup \{v\}$. It is easy to see that
$D\subseteq V(\cH)$, because if $v$ is in $D$, since $N_G[v]\cap \{y_e
    \mid e\in \cE(\cH)\}=\emptyset$ and $\cH$ contains at least one
non-empty hyperedge, $D$ must contain another vertex $x$ from $V(\cH)$.
But, since $N_G[v] \subseteq N_G[x]$, $P_D(v)=\emptyset$ which
contradicts the minimality of $D$ (cf. Lemma \ref{lem:2.1}).  We now
show that such a $D$ is a transversal of $\cH$.  Indeed since $D$ is
included in $V(\cH)$, every vertex $y_e$ with $e \in \cE(\cH)$ must be
incident with a vertex in $D\cap V(\cH)$ and then $D$ is a transversal
of $\cH$. Lemma \ref{lemme3} ensures that $D$ is a minimal transversal.
\end{proof}

\begin{thm}\label{theorem}
    \transhyp\ $\preceq_{op}$ \dom(co-bipartite graphs).
\end{thm}
\begin{proof}
    Assume there exists an output-polynomial time algorithm $\mathcal{A}$
    for the \dom\ problem which, given a co-bipartite graph $G$, outputs
    $\cD(G)$ in time $p(||G||+|\cD(G)|)$ where $p$ is a
    polynomial. Given a hypergraph $\cH$, we construct the co-bipartite
    graph $\cB(\cH)$ and call $\mathcal{A}$ on $\cB(\cH)$. By Lemma
    \ref{lemma2}, $\mathcal{A}$ on $\cB(\cH)$ outputs all minimal
    transversals of $\cH$.  We now discuss the time complexity. We
    clearly have $||\cB(\cH)||=O(||\cH||)$ and $\cB(\cH)$ can be
    constructed in time $O(||\cH||)$. Moreover by Lemma \ref{lemma2},
    $||\cD(\cB(\cH))||\leq ||tr(\cH)||+|V(\cH)|\times
    |\cE(\cH)|$. Therefore, $\mathcal{A}$ on $\cB(\cH)$ runs in time
    $O(p(||\cH||+||tr(\cH)||) + |V(\cH)|\times |\cE(\cH)|)$, which is
    polynomial on $||\cH||+||tr(\cH)|| $.
   \end{proof}


Corollary \ref{cor:leq}  and Theorem \ref{theorem} together imply the following result.
\begin{cor}\label{cor:equivalence}
    \dom(co-bipartite graphs), \dom\ and \transhyp\ are all equivalent.
\end{cor}



From Corollary \ref{cor:equivalence}, we can deduce some equivalences
between \dom\ and some other enumeration problems. For instance, a
\emph{total dominating set} is a dominating set $D$ such that the
subgraph induced by $D$ contains no isolated vertex. We call \TDS\ the
enumeration problem of (inclusion-wise) minimal total dominating
sets. To prove the next lemma we associate with every hypergraph a
split-incidence graph.

\begin{defn}\label{split-Inci}
	The split-incidence graph $\cI'(\cH)$ associated with
	a hypergraph $\cH$ is the graph obtained from $\cI(\cH)$ by
	turning the independent set corresponding to $V(\cH)$ into a clique
	(see Figure \ref{fig:IH}). The resulting graph is a split graph. 
\end{defn}
\begin{figure}[h!] 
\centering \includegraphics[scale=1.3]{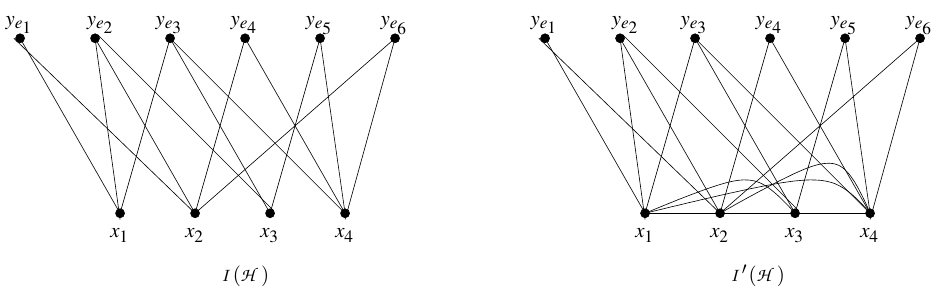} 
\caption{An example of
	the bipartite incidence graph $\cI(\cH)$ and the split-incidence graph
	$\cI'(\cH)$  of the hypergraph in Figure \ref{fig:3.1}.} \label{fig:IH}
\end{figure}

\begin{lem}\label{lem:tds}
    \TDS(split graphs),  \transhyp\ and \TDS\ are all equivalent.	
\end{lem}

\begin{proof}
It is enough to prove that \TDS\ $\preceq_{op}$ \transhyp\  and  \transhyp\ $\preceq_{op}$\TDS(split graphs) since  \TDS(split graphs)  $\preceq_{op}$ \TDS.
	We first show that \TDS\ $\preceq_{op}$ \transhyp\ (the reduction
	was first noted by \cite{thomasse07}).  For a graph $G$, we let
	$\cN_o(G):=(V(G),\{ N_G(x)\mid x\in V(G)\})$, the \emph{open
	    neighbourhood hypergraph}.  We claim that $\mathcal{TD}(G)=tr(\cN_o(G))$
	where $\mathcal{TD}(G)$ denotes the set of minimal total dominating sets
	of $G$. Let $G$ be a graph. It is easy to  see that
	$D\subseteq V(G)$ is a total dominating set in $G$ if and only
	if it is a transversal of $\cN_o(G)$. Indeed, if $D$ is a total
	dominating set of $G$, then for each $x\in V(G)$, $N_G(x)\cap
	D \ne \emptyset$. Therefore, $D$ is a transversal of $\cN_o(G)$.
	Conversely, if $T$ is a transversal of $\cN_o(G)$, then for each
	$x\in V(G)$, $T\cap N_G(x)\ne \emptyset$, \ie, $T$ is a total
	dominating set of $G$.

  We now show that \transhyp\ $\preceq_{op}$\TDS(split graphs).  Let $\mathcal{H}$
  be a hypergraph. Assume furthermore that $\mathcal{H}$ has no
  dominating vertex, \ie, a vertex belonging to all edges. Note that
  this case is not restrictive since if $x\in V(\cH)$ is a dominating
  vertex, then $tr(\mathcal{H})=\{x\}\cup tr(\mathcal{H}\setminus
  \{x\})$ and we can consider  this reduced hypergraph. We now show that
  $\mathcal{TD}(\cI'(\mathcal{H}))=tr(\mathcal{H})$.

  (i) Let $D$ be a minimal total dominating set of $\cI'(\cH)$, and
  let $e\in \cE(\cH)$. Then there exists $x\in V(\cH)\cap D$ such that
  $xy_e\in E(\cI'(\cH))$, \ie, $x\in e$.  We now claim that $y_e\notin D$
  for all $e\in \mathcal{E}(\cH)$. Otherwise, there exists $x\in e\cap D$
  and since $\cI'(\cH)[V(\cH)]$ is a clique, $D\setminus \{y_e\}$ is also
  a total dominating set, contradicting the minimality of $D$. Thus $D$
  is a transversal of $\cH$.

  (ii) Let $T$ be a transversal of $\cH$. Then for all $e\in \cE(\cH)$,
  $T\cap e\ne \emptyset$, \ie, for all $z\in V(\cI'(\cH))\setminus
  V(\cH)$, there exists $x\in T$ such that $xz\in E(\cI'(\cH))$. Since
  there is no dominating vertex, $|T|\geq 2$, and because
  $\cI'(\cH)[V(\cH)]$ is a clique, for all $x\in V(\cH)$,
  there exists $y\in T$ such that $xy\in E(\cI'(\cH))$. Hence, $T$ is
  a total dominating set of $\cI'(\cH)$.

  From (i) and (ii) we can conclude that
  $\mathcal{TD}(\cI'(\mathcal{H}))=tr(\mathcal{H})$.
\end{proof}
As a corollary of Lemma \ref{lem:tds} and Corollary
\ref{cor:equivalence} we get the following.
\begin{cor}
\dom\ and \TDS\ are equivalent. 
\end{cor}

These results  may enable new approaches to consider the
{\sc Trans-Enum}  problem as a graph problem. We will give some
evidence in the following sections. We conclude this section
	by stating the following decision problem \textsc{Dom-Graph} that arises
from Corollary \ref{cor:equivalence} and seems to be interesting on
	its own. \\\\\parbox{14cm}{
\textbf{Input.} A hypergraph $\cH$ and a positive integer $k$.\\
\textbf{Output.} Does there exist a graph $G$ and a set $F\subseteq
2^{V(G)}$ with $|F|\leq k$ \\ and such that
$\cD(G)= tr(\cH)\cup F$?\\\\}

It is an NP-complete problem because the problem of
\emph{realisability} of a hypergraph is a special case with $k=0$
\cite{BorosGZ08}. For $k=|V(\cH)|\cdot |\cE(\cH)|$, the
\textsc{Dom-Graph} problem can be solved in polynomial time by Corollary \ref{cor:equivalence}. We leave open its complexity for $1\leq k
< |V(\cH)|\cdot |\cE(\cH)|$.

\section{\dom\ in Split Graphs}\label{sec:4}
  
We recall that a graph $G$ is a \emph{split} graph if its vertex set can
be partitioned into an independent set $S$ and a clique $C$. Here we
consider $S $ to be maximal. We will denote a split graph $G$ by
the pair $(C(G)\cup S(G), E(G))$. We prove in this section that \dom(split graphs) admits a linear delay algorithm that uses polynomial space.
A minimal dominating set $D$ of a split graph $G$ can be partitioned
into a clique and an independent set, denoted respectively by $D_C:=D\cap
C(G)$ and $D_S:=D\cap S(G)$.  Lemma \ref{cliq} shows that a minimal
dominating set $D$ of a split graph is characterised by $D_C$. Note that
$D_S$ cannot characterise $D$, since several minimal dominating sets
can have the same set $D_S$. 
 
\begin{lem} \label{cliq}
     Let $G$ be a split graph and $D$ a minimal dominating set of
     $G$. Then $D_S=S(G)\setminus N_G(D_C)$.
\end{lem}


\begin{lem}\label{split-priv}
    Let    $A \subseteq C(G)$. If every element in $A$ has a private neighbour then   $A \cup (S(G) \setminus N_G(A))$ is a minimal
    dominating set of $G$.
\end{lem}
\begin{proof}
It is clear that $A\cup (S(G)\setminus N_G(A))$ is a dominating set. To see why it is minimal, it suffices to observe that every vertex $s\in S(G)\setminus N_G(A)$ has at least one private neighbour, namely vertex $s$ itself. Note that every vertex in $A$ has a private neighbour by assumption. Hence, $A\cup (S(G)\setminus N_G(A))$ is a minimal dominating set due to Lemma \ref{lem:2.1}.
\end{proof}

\begin{lem}\label{dom_split}
    Let $D$ be a minimal dominating set of a split graph $G$.  Then for all
    $A \subseteq D_C$, the set $A \cup (S(G) \setminus N_G(A))$ is a minimal
    dominating set of $G$.
\end{lem}

\begin{proof}

    Let $D$ be a minimal dominating set of $G$ and let $A \subseteq D_C$. 
    Clearly each $x\in A$ has a private neighbour since $D$ is a minimal dominating set. According to Lemma \ref{split-priv}, $A \cup (S(G) \setminus N_G(A))$ is a minimal
    dominating set of $G$.
\end{proof}

%

A consequence of Lemmas \ref{cliq} and \ref{split-priv} is the following. 
\begin{cor} \label{cor} Let $G$ be a split graph. Then  there is a
  bijection between $\cD(G)$ and the set $\{A\subseteq C(G)  \mid \forall x\in A, x \text{ has a private neighbour}
\}$. 
\end{cor}
%

%

We now describe an algorithm, which we call $DominantSplit$, that takes as input a split graph $G$ with a linear ordering $\sigma:V(G)\to \{1,\ldots, |V(G)|\}$ of its vertex set and a minimal
dominating set $D$ of $G$, and outputs all minimal dominating sets $Q$ of $G$ such that $D_C\subseteq Q_C$. Then, whenever $D=S(G)$, the algorithm enumerates all minimal dominating sets of $G$.  The algorithm starts by computing the largest vertex $y$ (with respect to the linear ordering $\sigma$) in $D_C$. Then, the algorithm checks whether the set
$D_C$ can be extended, i.e. whether there exists a vertex $x \in C(G)\setminus D_C$ which is greater than $y$ and such that every vertex in $D_C\cup\{x\}$ has a private neighbour. For
each such $x$, the algorithm builds the minimal dominating set $D'$ such that $D'_C=D_C\cup \{ x \}$ (which is unique by Lemma \ref{cliq}) and recursively calls the algorithm on $D'$. The pseudo-code
is given in Algorithm \ref{algo:split}.

\begin{algorithm2e}[h!]
    \dontprintsemicolon
\SetVline
\caption{$DominantSplit(G, \sigma,D)$}
\SetKw{dominant}{DominantSplit}
\SetKw{aoutput}{output}
\SetKw{Cov}{Cov}
\SetKw{nCov}{NewCov}
 \label{dominantsplit}
\KwIn{A split graph $G=(C(G)\cup S(G),E(G))$, a linear ordering $\sigma:V(G)\to \{1,\ldots, |V(G)|\}$ and a minimal dominating
    set $D$ of  $G$.}
\Begin{
\aoutput(D)\;
$\Cov=\emptyset$\;
\nl Let $y\in D_C$ be such that $\sigma(y) =max\{\sigma(x)\mid x\in D_C\}$\; 
\nl \ForEach{$x\in  C(G) \setminus D_C$ and $ \sigma(x) >\sigma(y)$}{
\nl \If{each vertex in  $D_C\cup \{x\}$  has a private neighbour}{
\nl  $\Cov=\Cov\cup \{x\}$\;
}}
\nl  \ForEach{ $x \in  \emph{\Cov}$}{
    \nl $DominantSplit(G, \sigma, D_C\cup \{x\}\cup  (S(G)\setminus
    N_G(D_C\cup\{x\}))$) \;
 }}
\label{algo:split}
\end{algorithm2e}

\begin{thm}\label{delay} Let $G$ be a split graph with $n$ vertices
  and $m$ edges and let $\sigma$ be any linear ordering of $V(G)$. Then $DominantSplit(G,\sigma,S(G))$
    enumerates the set $\cD(G)$ with $O(n+m)$ delay and uses space
    bounded by $O(n^2)$.
\end{thm}
\begin{proof}
   We first prove the correctness of the algorithm. We first prove that
   each minimal dominating set is listed once.
   
   We prove the completeness  using induction on the number of elements in the clique $C(G)$. First the only minimal dominating  set
    $D$ of $G$ such that $|D_C|=0$ is $S(G)$ which corresponds to the
   first call of the algorithm. Indeed, if $D\cap C=\emptyset$ then each
   vertex of $S(G)$ must belong to $D$ to dominate itself.  Moreover,
   by Lemma \ref{split-priv} it
   is a minimal dominating set. Assume now
   that every $D'\in \cD(G)$ such that $|D'_C|\leq k$, is returned by the
   algorithm and let $D$ be a minimal dominating set such that $|D_C|=k+1$.
   Let $x$ be the greatest vertex of $D_C$ (with respect to $\sigma$). By Lemma \ref{dom_split},
   $D':=D_C \setminus\{ x \}\cup(S(G)\setminus N_G(D_C\setminus \{ x \}))$ is
   a minimal dominating set of $G$. Furthermore, $|D'_C|=|D_C \setminus
   \{ x \}|=k$, and then $D'$ is returned by the algorithm (by the inductive hypothesis). Note also
   that since $D\in \cD(G)$, every vertex in $D'_C \cup \{ x \}=D_C$ has
   a private neighbour, and since $x$ is greater than all vertices
   of $D'_C$ (w.r.t. $\sigma$), $x$ is added to  \textbf{Cov} by the algorithm. Then in the next
   step $DominantSplit(G,\sigma, D'_C\cup \{x\}\cup  (S(G)\setminus
   N_G(D'_C\cup\{x\}))))$ will be called and then $D'_C\cup \{x\}\cup  (S(G)\setminus
   N_G(D'_C\cup\{x\}))= D_C \cup (S(G)\setminus  N_G(D_C)) $ will be
   returned, which is equal to $D$ by Lemma \ref{cliq}.
   
   Now let us show that if  a set $A$ is returned, then $A$ is a minimal dominating set of $G$. We have two cases:
   either $A=S(G)$ (which corresponds to the first call) or $A=D_C\cup \{x\}\cup  (S(G)\setminus N_G(D_C\cup\{x\}))$. 
   Clearly if  $A=S(G)$ then $A$ is a minimal dominating set. Now if
   $A=D_C\cup \{x\}\cup  (S(G)\setminus N_G(D_C\cup\{x\}))$ then every
   element in $D_C\cup \{x\}$ has a private neighbour (cf.  Line 3 of
   algorithm \ref{dominantsplit}). Using Lemma \ref{split-priv} we conclude that $A$ is a minimal dominating set.
   
   Moreover each minimal dominating set is listed exactly
   once. Indeed, a minimal dominating set $D'$ is obtained by a call
   $DominantSplit(G,\sigma, D_C\cup \{x\}\cup (S(G)\setminus
   N_G(D_C\cup\{x\})))$ where $D$ is the minimal dominating set such
   that $D_C= D'_C\setminus \{x\}$ with $x$ the greatest
   vertex in $D'_C$ (with respect to $\sigma$), which is unique by Lemma
   \ref{cliq}.

   We now	discuss the delay and space. The delay between the output of $D$ and the
	next output is dominated by the time needed to check if any element in
	$D_C\cup \{x\}$ has a private neighbour. 
	
	To do so, we use an array marks[1..n] initialised to $0$, and
	for each element in $D_C\cup \{ x \}$ we increase  the marks of its
	neighbours by $1$. To check that every element $y$ in $D_C\cup
	\{x\}$ has a private neighbour, it suffices to check that $y$
	has at least a neighbour with mark $1$. Note that we check only neighbourhood in the stable $S(G)$.  This can be done in
	time $O(n+m)$.  
	Since the depth of the recursive tree is at most
	$n$ and at each node we store the set {\bf Cov}, the space
	memory is bounded by $O(n^2)$. 
\end{proof}

\section{Completion}\label{sec:6}

In this section we introduce the notion of the maximal extension of a graph by
keeping the set of minimal dominating sets invariant. The idea behind this
operation is to maintain invariant the minimal  hyperedges, with respect to
inclusion, in $\cN(G)$.

For a graph $G$ we denote by $IR(G)$ 
the set of
vertices (called \emph{irredundant} vertices) that are minimal with respect to the neighbourhood inclusion. In
case of equality between minimal vertices, exactly one is considered as
irredundant.  All the other vertices are called \emph{redundant} and
the set of redundant vertices is denoted by $RN(G)$.
The \emph{completion} graph of a graph $G$ is the graph $G_{co}$ with
vertex set $V(G)$ and edge set $E(G)\cup \{xy\mid x,y\in RN(G),~ x\not
= y\}$, \ie,  $G_{co}$ is obtained from $G$ by adding precisely those edges to $G$ that make   $RN(G)$ into a clique.
%
%
Note that the completion graph of a split graph
$G$ is $G$ itself, since all vertices in $S(G)$ are irredundant. However, the completion operation does not preserve
the chordality of a graph. For instance, trees are chordal graphs but
their completion graphs are not always chordal. Figure \ref{fig:5.1}
gives some examples of completion graphs.

\begin{rem}\label{rem:1} Note that if a vertex $x$ is
redundant, then there exists an irredundant vertex $y$ such that
$N_G[y]\subseteq N_G[x]$. Indeed since $x$ is redundant, the set
$F:=\{ z\in V \mid N_G[z]\subseteq N_G[x] \}$ is not empty. Hence, any
minimal (with respect to neighbourhood inclusion) vertex $y$ from $F$
is an irredundant vertex.
\end{rem}

\begin{figure}[h!]
\begin{center}
\includegraphics[scale=.85]{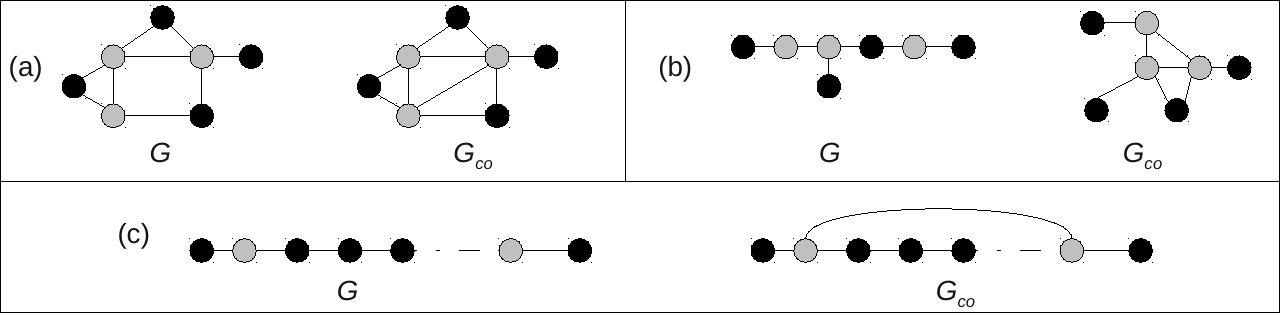}
\end{center}
\caption{(a) a non-chordal graph whose completion is a split graph (b)
    a chordal graph with an induced $P_6$ whose completion is a split
    graph (c) a path $P_n$ whose completion is not chordal.  Redundant
    vertices are represented in grey.} \label{fig:5.1}
\end{figure}

\begin{prop} \label{lem:5.1} For any graph $G$, we have
	$\mathcal{D}(G)=\mathcal{D}(G_{co})$.
\end{prop}

\begin{proof}
	Let $D$ be a dominating set of a graph $G$. Since  $E(G) \subseteq
	E(G_{co})$, $D$ is also a dominating set of $G_{co}$.  Now suppose that
	$D$ is a dominating set of $G_{co}$ and let $x\in V(G)$. If $x\in IR(G)$,
	then $N_{G}[x]=N_{G_{co}}[x]$, hence $D\cap N_{G}[x]\neq \emptyset$.  If
	$x\in RN(G)$, then, due to Remark \ref{rem:1}, there exists $y\in IR(G)$ such that $N_G[y]\subseteq
	N_G[x]$. Hence $D\cap N_{G}[y] \subseteq D\cap N_{G}[x] \neq \emptyset$.
	Therefore, $D$ is a dominating set of $G$. {Since $G$ and $G_{co}$ have
		the same dominating sets, we deduce that $\cD(G)=\cD(G_{co})$}.
\end{proof}


The following proposition claims the optimality of the completion in
the sense that no other edges can be added to the graph without
changing the set of minimal dominating sets.

\begin{prop}
Let $G$ be a graph and let $G'$ be $(V(G), E(G)\cup \{e\})$ with $e$
a non-edge of $G$. Then  $\cD(G)\neq \cD(G')$  if and only if $e \cap IR(G)\neq \emptyset$, 
\end{prop}
\begin{proof}
Consider
 $\mathcal{N'}(G):=\{N_{G}[v]\mid v\in IR(G)\}$ and
 $\mathcal{N'}(G'):=\{N_{G'}[v]\mid v\in IR(G')\}$. By the definition of irredundant vertices, for every $u,v\in IR(G)$, we have $N_{G}[u] \subseteq N_{G}[v]$ implies that $u=v$ and therefore $N_{G}[u] = N_{G}[v]$. Hence  $\mathcal{N'}(G)$ and $\mathcal{N'}(G')$ are simple and correspond respectively to $Min(\mathcal{N}(G))$ and $Min(\mathcal{N}(G'))$.
 
 Let $e:=xy$  such that $e\cap IR(G)\neq \emptyset$, and assume without loss of generality that $x\in IR(G)$.  Assume that $x$ is still
 irredundant in $G'$, \ie, $x\in IR(G')$. Then since $y\in N_{G'}[x]$ and
 $y\notin N_{G}[x]$, $\mathcal{N'}(G')\neq \mathcal{N'}(G)$.  Moreover, thanks to Lemma \ref{lem:3.1},  we
 have $\cD(G)=tr(\cN'(G))$ and $\cD(G')=tr(\cN'(G'))$, and since $\cN'(G)$ and
 $\cN'(G')$ are simple, we have $\cD(G')\neq\cD(G)$ (see Proposition \ref{prop:2.1}).  Assume now that $x\in
 RN(G')$. Hence, $N_G[x]\notin \cN'(G')$ and since $N_G[x]\in \cN'(G)$, we
 have $\cN'(G)\neq \cN'(G')$. 
 
 Assume now that $e \cap IR(G)= \emptyset$, i.e. $e\subseteq RN(G)$.
 Then $IR(G)=IR(G')$ and for all $v\in IR(G)$, $N_{G}[v]=N_{G'}[v]$. Thus we have  $Min(\cN(G))=\cN'(G)=\cN'(G')=Min(\cN(G'))$ and then   $\cD(G)=tr(\cN(G))=tr(\cN(G'))=\cD(G')$.  
\end{proof}

We now  show how to  use completion to get an output-polynomial time
algorithm for the \dom\ problem restricted to $P_6$-free chordal graphs.  Let us
notice that this class properly contains the class of split graphs. The results that
follow were  already published in \cite{KLMN11} without proofs. A vertex is \emph{simplicial} if the graph induced by its neighbourhood is a
clique.

\begin{prop}\label{lemme_P_6}
	If $G$ is a  $P_6$-free chordal graph, then for all $x\in IR(G)$,
	$x$ is a simplicial vertex in $G_{co}$.  Furthermore, the set
	$IR(G)$ is an independent set in $G_{co}$.
\end{prop}

\begin{proof} 
	We first show that  for all $x\in IR(G)$, $x$ is a simplicial vertex in
	$G_{co}$.  Assume that there exists  $x\in IR(G)$ such that $x$ is
	not a simplicial vertex in $G_{co}$. Then there exist $y,z\in
	N_{G_{co}}[x]$ such that $yz\notin E(G_{co})$.  Since $x$ is irredundant
	in $G$, there exist $y'\in N_{G}[y]\setminus N_{G}[x]$ and $z'\in
	N_{G}[z]\setminus N_{G}[x]$. Observe that $y'\neq z'$, otherwise
	$\{x,y,y',z\}$ forms an induced $C_4$ in $G$. Moreover, since $yz\notin
	E(G_{co})$, either $z\notin RN(G)$ or $y\notin RN(G)$. Assume without loss
	of generality that $y\notin RN(G)$. Then $N_{G}[y']\not\subseteq
	N_{G}[y]$ and so there exists $y'' \in N_{G}[y']\setminus N_{G}[y]$. But
	then $P:=z'zxyy'y''$ forms an induced $P_6$, because all possible
	edges between two non consecutive vertices of $P$ would create an induced
	cycle of length greater than four, contradicting the chordality of $G$.

  We finally show that $IR(G)$ is an independent set in $G_{co}$.  Suppose that
  there exists $xy\in E(G_{co})$ with $x,y\in IR(G)$. Since for all $z\in
  IR(G)$, $z$ is a simplicial vertex in $G_{co}$, it follows that both
  $N_{G_{co}}[x]$ and $N_{G_{co}}[y]$ are cliques. But since $xy\in
  E(G_{co})$, we have $N_{G_{co}}[x]=N_{G_{co}}[y]$, otherwise there
  must exist $z\in N_{G_{co}}[x]\setminus y$ and $yz\notin E(G_{co})$
  which is impossible since  $x$ is simplicial (by the first
  statement).  
Since no edges are added incident with $x$ or $y$ when $G_{co}$ is obtained
from $G$, we must have $N_{G}[x]=N_{G}[y]$ contradicting the
assumption that $x$ and $y$ are irredundant.
\end{proof}

A consequence of Proposition \ref{lemme_P_6} is the following.

\begin{prop} \label{P6free} 
	Let $G$ be a $P_6$-free chordal graph. Then  $G_{co}$ is a split graph. 
\end{prop}
\begin{proof}
	From Proposition \ref{lemme_P_6}, it follows that $IR(G)$ forms an
	independent set in $G_{co}$, and since $RN(G)$ forms a clique in $G_{co}$, we are
	done.
\end{proof}
The next theorem characterises completion graphs that are split. 

\begin{prop}\label{thm:5.1} 
	Let $G$ be a graph. Then  $G_{co}$ is a chordal graph if and only if
	$G_{co}$ is a split graph.
\end{prop}

\begin{proof} Since split graphs are chordal graphs, it is enough to
  prove that if $G_{co}$ is chordal, then it is a split graph. Assume
  there exists a graph $G$ such that $G_{co}$ is chordal and not a split
  graph. 
Since $RN(G)$ forms a clique in $G_{co}$, there must exist $x_1,x_2 \in
IR(G)$ such that $x_1x_2 \in E(G_{co})$. We prove the following claim,
which contradicts the fact that $G$ is finite and therefore suffices to
prove Proposition \ref{thm:5.1}. 
 
  \begin{Claim}\label{claim:5.1} 
	  There exists an infinite sequence $(x_i)_{i\in \bN}$ of distinct
	  vertices in $IR(G)$ such that, for all $i$, $x_i$ is connected to
	  $x_{i+1}$ and $x_{i-1}$, and for $j\notin\{i-1,i+1\}$, $x_ix_j\notin
	  E(G_{co})$.
  \end{Claim}

  \begin{proof}[Proof of Claim \ref{claim:5.1}] 
	  Since $x_1\in IR(G)$, there exists $x_2' \in N_{G}[x_2]\setminus
	  N_{G}[x_1]$. In the same way, there exists $x_1'\in N_{G}[x_1]\setminus
	  N_{G}[x_2]$. \\
	\textbf{Case 1.}~ $x_1'\in RN(G)$ and $x_2'\in RN(G)$.  Then $x_1'x_2'\in
	E(G_{co})$ and so $C:=x_1x_2x_2'x_1'$ forms an induced $C_4$ of
	$G_{co}$, which contradicts the assumptions.\\
	\textbf{Case 2.}~ $x_1' \in RN(G)$ and $x_2'\in IR(G)$.  Let $x_3=x_2'$.  We prove
	by induction that for all $j\geq 3$, there exists an induced path
	$x_1\ldots x_{j}$ of elements of $IR(G)$.  For $j=3$ the property holds since
	$x_1x_2x_3$ forms an induced path.
      
	Assume that the property holds for all $j\leq k$, in other
	words, we have a sequence $P:=x_1x_2\ldots x_k$ of $k$ distinct
	elements of $IR(G)$ forming an induced path in $G_{co}$. We show
	now that there exists $x_{k+1}\in IR(G)$ such that $x_{k+1}x_k\in
	E(G_{co})$ and for all $j \leq k$, $x_{k+1}x_j\notin E(G_{co})$.
	Since $x_{k-1}\in IR(G)$, there exists a vertex in
	$N_{G}[x_k]\setminus N_{G}[x_{k-1}]$. 
	We choose $x_{k+1}$ to be such a vertex. Note that
	$x_{k+1}\notin \{ x_k,x_{k-1}  \}$ since $x_{k+1}\in
	N_G[x_k]\setminus N_G[x_{k-1}]$, by definition, and
	$x_{k+1}\notin\{ x_1,...,x_{k-1}  \}$ since $x_1\ldots x_k$
        is an induced path and $x_{k+1}$ is adjacent to $x_k$. In 
	other words $x_{k+1}$ is distinct from $x_j$ for all $j\leq k$.
	Also note that $x_{k+1}$ cannot belong to $RN(G)$, as otherwise
	$x_{k+1}x_1'x_1\ldots x_k$ would be a cycle of
	length greater than four, contradicting the assumption that $G_{co}$
	is chordal. Since $x_{k+1}\in N_{G_{co}}[x_k]\setminus
	N_{G_{co}}[x_{k-1}]$, if there exists $j<k-1 $ with
	$x_jx_{k+1}\in E(G_{co})$ then $x_j \ldots x_{k+1}$
	induces a cycle of length at least four. This contradiction
	finished the proof of Case 2.

	\textbf{Case 3.}~ $x_1' \in IR(G)$ and $x_2'\in IR(G)$.  Case 3 is identical to
	Case 2 up to symmetry.
  \end{proof}
  \end{proof}


We can now state the following theorem which generalises Theorem
\ref{delay} to $P_6$-free chordal graphs. Actually, $P_6$-free chordal graphs properly  contain  split graphs, since split graphs are $P_5$-free chordal graphs.

\begin{thm} \label{thm:P6}
	There exists an $O(n+m)$ delay algorithm for the \dom\ problem in
	$P_6$-free chordal graphs with space complexity $O(n^2)$. 
\end{thm}
\begin{proof}

Let $G$ be a $P_6$-free chordal graph. First, construct the graph
$G_{co}$, which can clearly be done in polynomial time. Then, enumerate
all minimal dominating sets of $G_{co}$, which can be done with linear delay  in the size of $G$ (since the added edges in the completion  are not considered  by Algorithm \ref{algo:split}) and
 using $O(n^2)$ space due to Theorem \ref{delay}. The observation that this set coincides with the set of all minimal dominating sets of $G$ due to Proposition \ref{lem:5.1} finishes the proof of Theorem \ref{thm:P6}.
\end{proof}

%


\section{Connected Dominating Sets}\label{sec:7}

We investigate in this section the complexity of the enumeration of
minimal \emph{connected dominating sets} of a graph. A \emph{connected
  dominating set} is a dominating set $D$ such that the subgraph
induced by $D$ is connected; it is minimal if for each $x\in D$,
either $D\setminus \{x\}$ is not a dominating set or the subgraph 
induced by $D\setminus \{x\}$ is not
connected.  We denote by \CDS\ the enumeration problem of minimal
connected dominating sets, and  by $\cC\cD(G)$ the set of minimal
connected dominating sets of a graph $G$.

\begin{prop}[\cite{KLMN11}]\label{prop:6.1} For every hypergraph $\cH$,
    $tr(\cH)=\cC\cD(\cI'(\cH))$. Hence, \CDS(split graphs) is equivalent to \transhyp.
\end{prop}
\begin{proof}

  (i) Let $D\in \cC\cD(\cI'(\mathcal{H}))$ (cf. Definition \ref{split-Inci}). 
  Note   that  every minimal connected dominating set in a split graph is a subset of the clique (cf. \cite{Babel98})  and thus $D\subseteq V(\cH)$. 
  Now, for each $e\in \cE(\cH)$,
  there exists $x\in D$ such that $xy_e\in E(\cI'(\cH))$, hence
  $D\cap e\ne \emptyset$. And so $D$ is a transversal of $\cH$. 

  (ii) Let $T$ be a transversal of $\cH$. Since $\cI'(\cH)[V(\cH)]$ is a
  clique, $T$ is connected, and for each $x\in V(\cH)$, there exists
  $y\in T$ such that $xy\in E(\cI'(\cH))$. Furthermore, for each $e\in
  \cE(\cH)$, $T\cap e\neq \emptyset$, \ie, for each $y_e\in
  V(\cI'(\cH))\setminus V(\cH)$, there is $z\in T$ such that $zy_e\in
  E(\cI'(\cH))$. Hence, $T$ is a connected dominating set of
  $\cI'(\cH)$.

  From (i) and (ii) we can conclude that
  $\cC\cD(\cI'(\mathcal{H})) = tr(\mathcal{H})$.

  It remains to reduce \CDS\ to \transhyp. For a split graph $G$, we let
  $\cH$ be the hypergraph $(C(G), \{N_G(x)\mid x\in S(G)\})$. It is
  easy to see that $G=\cI'(\mathcal{H})$ and so from above,
  $\cC\cD(\cI'(\mathcal{H}))=tr(\mathcal{H})$.
\end{proof}

We will extend this result to other graph classes and we expect that
it is a first step for classifying the complexity of the \CDS\
problem. 

A subset $S\subseteq V(G)$ of a connected graph $G$ is called a \emph{separator}
of $G$ if $G\setminus S$ is not connected; $S$ is minimal if it does
not contain any other separator.  
Note that this notion is different from
the classical notion of minimal \emph{$ab$-separators}.  For two
vertices $a$ and $b$, an \emph{$ab$-separator} is a subset $S\subseteq
V(G)\setminus\{a,b\}$ which disconnects $a$ from $b$; it is said to be
minimal if no proper subset of $S$ disconnects $a$ from $b$.  
Every minimal separator is an $ab$-separator for some pair of vertices $a, b$.
The minimal separators are exactly the minimal $ab$-separators which do
not contain any other $cd$-separator. For this reason they are often
called the \emph{ inclusion minimal separators}. Notice that a graph
may have an exponential number of minimal separators, but one can
enumerate them  in output-polynomial time \cite{SL97}.   Algorithms that enumerate all the minimal $ab$-separators of a graph
 can be found in  \cite{BerryBC00,KloksK98,SL97}. We define $\mathcal{S}(G)$ as the hypergraph $(V(G),\{S\subset V(G) \mid S \text{ is a minimal separator of } G\})$.

\begin{prop}\label{separator} For every graph $G$, $\cC\cD(G)=tr(\cS(G))$.
\end{prop}

\begin{proof}
We first prove that a connected dominating set of $G$ is a transversal
of $\cS(G)$.  Let $D$ be a connected dominating set and assume that
there exists a separator $S$ for which $S\cap D=\emptyset$. Let $G_1,
\ldots, G_p$ be the connected components of $G[V\setminus S]$. Since
$D$ is connected, it must be included in $V(G_i)$ for some $1\leq
i\leq p$. Assume without loss of generality that $D\subseteq V(G_1)$
and let $x\in V(G_2)$. Then we have $N_G[x]\subseteq V(G_2)\cup S$ and
then $N_G[x]\cap D=\emptyset$ which contradicts the fact that $D$ is a
dominating set of $G$.

We now prove that a transversal of $\cS(G)$ is a connected dominating set of
$G$. Let $T$ be a transversal of $\cS(G)$. We first show that $T$ is a
dominating set of $G$. Suppose not and let $N$ be the set of vertices not
covered by $T$, \ie, $N:=\{x\in V(G)\mid N_G[x]\cap T=\emptyset\}$. Then
$V(G)=T\cup N_G(T)\cup N$ and by definition of $N$, there are no edges between
$N$ and $T$. So $G\setminus N_G(T)$ is not connected, in other words, $N_G(T)$
is a separator of $G$. Hence, $N_G(T)$ contains a minimal separator $S$ which
does not intersect $T$. This  contradicts the fact that $T$ is a transversal of
$\cS(G)$. It remains to prove that $G[T]$ is connected. Assume, for contradiction,
that $G[T]$ is  not connected.  Then $V(G)\setminus T$ is a separator. But then $V(G)\setminus T$ contains a minimal
separator $S$ such that $S\cap T \neq \emptyset$.
This contradicts again the fact that $T$ is a transversal of
$\cS(G)$.

Finally since a set $S$ is a transversal of $\cS(G)$ if and only if $S$ is a connected
dominating set of $G$, we have that $tr(\cS(G))=\cC\cD(G)$.  
\end{proof}

The following corollary shows that any simple hypergraph is the set of minimal
separators for some graph, whereas there exist simple hypergraphs which are not
neighbourhood hypergraphs (see \cite{BorosGZ08}).

\begin{cor}
For each simple hypergraph $\cH$, there exists a split graph $G$ such that $\cH=\cS(G)$. 
\end{cor}
\begin{proof}  By Proposition
  \ref{separator}, we have $\cC\cD(\cI'(\cH)) = tr(\cS(\cI'(\cH)))$. So,
  by Proposition \ref{prop:6.1}, 
  we have $tr(\cH)=tr(\cS(\cI'(\cH)))$, and then by Proposition
  \ref{prop:2.1}, $\cH= \cS(\cI'(\cH))$.
\end{proof}
Another consequence of Proposition \ref{separator} is the following.
\begin{cor}
If a class of graphs $\cC$ has a polynomially bounded number of minimal
separators, then  \CDS($\cC$)\ $\preceq_{op}$ \transhyp. Moreover, if the class $\cC$ contains split graphs, then
\transhyp\ is equivalent to \CDS($\cC$).
\end{cor}
\begin{proof}
	Assume that one can solve \transhyp\ in output-polynomial time. Let $G\in
	\cC$. Since the set of all minimal separators of a graph can be enumerated
	in output-polynomial time and since there is a polynomial number of
	separators, $\cS(G)$ can be computed in time polynomial in
	$||G||$. Furthermore, the fact that $\cC\cD(G)=tr(\cS(G))$ by
	Proposition \ref{separator} achieves the proof of the first
        statement. The second statement follows from the first
        statement and Proposition \ref{prop:6.1}.   
\end{proof}
 Among examples of such graph classes we can cite, without being exhaustive, 
chordal graphs, trapezoid graphs \cite{BKKM98},  chordal bipartite graphs \cite{KK95},
and circle and circular arc graphs \cite{KKW98}. 
\bibliographystyle{plain}
\bibliography{bib}

\end{document}